\documentclass[twocolumn,english,aps,prd,nofootinbib,reprint,floatfix,notitlepage,preprintnumbers,superscriptaddress]{revtex4-1}
\pdfoutput=1
\usepackage{lmodern}

\usepackage[T1]{fontenc}
\usepackage[latin9]{inputenc}
\usepackage{geometry}
\usepackage{subfigure,lmodern, amsmath,amssymb, graphicx, pifont, adjustbox, bm}
\usepackage[dvipsnames]{xcolor}
\usepackage{amsfonts}
\usepackage{geometry}
\usepackage{enumitem}
\usepackage{comment}
\usepackage{mathtools}
\usepackage{float}
\usepackage{slashed}
\usepackage{ragged2e}
\usepackage{array}
\usepackage{amsthm}

\newtheorem{theorem}{Theorem}[section]
\newtheorem{corollary}[theorem]{Corollary}
\newtheorem{lemma}[theorem]{Lemma}
\newtheorem{proposition}[theorem]{Proposition}
\theoremstyle{definition}

\theoremstyle{remark}

\makeatletter\g@addto@macro\bfseries{\boldmath}\makeatother

\usepackage{stackengine}
\usepackage{esint}
\usepackage[unicode=true,pdfusetitle,
 bookmarks=true,bookmarksnumbered=false,bookmarksopen=false,
 breaklinks=false,pdfborder={0 0 1},backref=false,colorlinks=true]
 {hyperref}
\hypersetup{
 pdfauthor={Aidan Chatwin-Davies, Pompey Leung, and Grant N. Remmen},
 citecolor=black,linkcolor=black,urlcolor=black}

\newcommand{\appendixref}[1]{\hyperref[#1]{appendix~\ref{#1}}}
\def\equationautorefname~#1\null{eq.\,(#1)\null}
\usepackage{breakurl}
\usepackage{breakurl}
\usepackage[hang,flushmargin]{footmisc} 
\allowdisplaybreaks
\makeatletter

\usepackage{etoolbox}
\apptocmd{\thebibliography}{\justifying\setlength{\leftskip}{7.4mm}}{}{} 
 
\usepackage{relsize}
\usepackage{babel}

\makeatletter
\def\simgt{\mathrel{\lower2.5pt\vbox{\lineskip=0pt\baselineskip=0pt
           \hbox{$>$}\hbox{$\sim$}}}}
\def\simlt{\mathrel{\lower2.5pt\vbox{\lineskip=0pt\baselineskip=0pt
           \hbox{$<$}\hbox{$\sim$}}}}
\makeatother

\usepackage{changepage}

\makeatletter
\def\l@subsubsection#1#2{}
\makeatother

\newcommand{\be}{\begin{equation}}
\newcommand{\ee}{\end{equation}}
\newcommand{\bea}{\begin{eqnarray}}
\newcommand{\eea}{\end{eqnarray}}
\newcommand{\Fig}[1]{Fig.~\ref{#1}}

\newcommand{\Eq}[1]{Eq.~(\ref{#1})}

\newcommand{\Sec}[1]{Sec.~\ref{#1}}

\newcommand{\mrm}[1]{\mathrm{#1}}

\usepackage[normalem]{ulem}

\DeclareMathOperator{\tr}{tr}
\DeclareMathOperator*{\ext}{ext}

\newcolumntype{P}[1]{>{\centering\arraybackslash}p{#1}}

\begin{document}
\title{Holographic screen sequestration}

\author{Aidan Chatwin-Davies}
\affiliation{Okinawa Institute of Science and Technology\\ 1919-1 Tancha, Onna-son, Kunigami-gun, Okinawa, Japan 904-0495\\[1mm]}
\author{Pompey Leung}
\affiliation{Department of Physics and Astronomy, University of British Columbia, 6224 Agricultural Road, Vancouver, British Columbia, V6T 1Z1, Canada\\[1mm]}
\author{Grant N. Remmen}
\affiliation{Center for Cosmology and Particle Physics, Department of Physics, New York University, New York, New York 10003, United States}
    
\begin{abstract}
\noindent 
Holographic screens are codimension-one hypersurfaces that extend the notion of apparent horizons to general (non-black hole) spacetimes and that display interesting thermodynamic properties. 
We show that if a spacetime contains a codimension-two, boundary-homologous, minimal extremal spacelike surface $X$ (known as an HRT surface in AdS/CFT), then any holographic screens are sequestered to the causal wedges of $X$.
That is, any single connected component of a holographic screen can be located in at most one of the causal future, causal past, inner wedge, or outer wedge of $X$.
We comment on how this result informs possible coarse grained entropic interpretations of generic holographic screens, as well as on connections to semiclassical objects such as quantum extremal surfaces.

\end{abstract}

\maketitle

\section{Introduction}

The past quarter-century of the anti--de Sitter/conformal field theory (AdS/CFT) correspondence~\cite{Maldacena:1997re,Gubser:1998bc,Witten:1998qj,Aharony:1999ti} has ushered in a revolution of our understanding of gravitation and spacetime.
Rather than being fundamental concepts, in AdS/CFT both quantum gravity and the bulk dimension of spacetime itself can be seen as emergent phenomena~\cite{Susskind:1994vu,Bousso:2002ju,Faulkner:2013ica}, arising as effective descriptions encoded via complicated and subtle details of the dynamics in the nongravitational world of the boundary.
This encoding has been understood as a form of quantum error correction~\cite{Almheiri:2014lwa,Pastawski:2015qua}, with the geometry encoding entanglement---specifically, the fine grained (von~Neumann) entropy---through areas of surfaces, or more general quantities in a gravitational effective field theory~\cite{Dong:2013qoa}.\footnote{See Refs.~\cite{Dong:2023bax,Gesteau:2023hbq} for recent analyses of the relationship between renormalization and holographic entropy in effective theories.}
The Ryu-Takayanagi (RT) formula~\cite{Ryu:2006bv}, and more generally the Hubeny-Rangamani-Takayanagi (HRT) formalism~\cite{Hubeny:2007xt} for dynamical spacetimes, give the prescriptions for finding the appropriate bulk surfaces whose areas compute the entropy corresponding to a given boundary region, generalizing the celebrated Bekenstein-Hawking entropy of black holes~\cite{Bekenstein:1972tm,Bekenstein:1973ur,Hawking:1976de}.

Despite the success of these famous results and many others, significant puzzles remain. How does spacetime emerge when it is not anti--de Sitter? That is, what is the analogue of the boundary theory for general spacetimes, and where is it 
located?
As a starting point, one can look to history, where Hawking's area theorem for black holes arguably was a critical seed leading eventually to our current holographic understanding of AdS/CFT.
Happily, new area theorems for general spacetimes have been found by Bousso and Engelhardt~\cite{Bousso:2015mqa,Bousso:2015qqa}, in the form of  monotonic area increase along a special hypersurface called a {\it holographic screen}.
The holographic screen has been conjectured to be the appropriate analogue of the AdS boundary for a holographic description of general spacetimes~\cite{Bousso:1999cb,Bousso:2015mqa,Bousso:2015qqa,Nomura:2017npr}.
Within AdS/CFT, an entropic interpretation of the outermost spacelike portion of a holographic screen, that is, an apparent horizon, was provided by Refs.~\cite{Engelhardt:2017aux,Engelhardt:2018kcs}. 
There it was shown that the area of a slice of the screen is equal to the {\it outer entropy}: the area of the maximal HRT surface, in units of $4G\hbar$, subject to holding the geometry in the outer wedge---the exterior of the slice of the screen---fixed and marginalizing over all possible completions of the spacetime.
In this sense, at least the apparent horizon can be viewed as a coarse grained holographic measure of entanglement entropy.
Moreover, the outer entropy can be viewed as computing a quasilocal energy for general (i.e., nonmarginal) surfaces~\cite{Bousso:2018fou,Wang:2020vxc}.

To truly have a general formulation of holography for arbitrary spacetimes, however, it is likely that a necessary condition is an interpretation of the {\it entire} holographic screen in information-theoretic terms.
A construction for timelike holographic screens analogous to Refs.~\cite{Engelhardt:2017aux,Engelhardt:2018kcs} is as yet unknown.
In order to find such a construction, it would be useful to have a complete characterization of where the holographic screen and HRT surface can be located relative to each other.

It is this latter question that we answer in this paper. 
Specifically, we show that holographic screens are {\it sequestered} relative to the HRT surface: given an HRT surface, a holographic screen is contained entirely within either its interior, exterior, past, or future, being forever forbidden from crossing the null congruences launched from the HRT surface.
There are special cases in which, under certain conditions, the screen can touch---but not pass through---a null congruence from the HRT surface.
Specifically, while such instantaneous intersections are strictly forbidden for spherically symmetric geometries, in nonspherical cases they cannot be ruled out if either (i) some generators of the congruence do not intersect the screen (i.e., the intersection is incomplete), or (ii) the intersection is a surface of zero Euler characteristic. Possibility (i) is fairly generic, and indeed a complete intersection can always be infinitesimally deformed to an incomplete one. However, possibility (ii) is fairly exotic, as it requires that the HRT surface be, e.g., a topological torus, Klein bottle, or odd-dimensional sphere.

This paper is organized as follows.
After providing a self-contained primer of the geometric background and definitions in Sec.~\ref{sec:BackgroundDefs}, we prove our main sequestration result, Theorem~\ref{thm:Sequestration}, in Sec.~\ref{sec:Sequestering}.
We enumerate the possibilities for exceptional cases where the screen and null congruences can touch in detail in Sec.~\ref{ssec:EdgeCases}.
We comment on obstructions to straightforward application of a coarse grained interpretation of timelike holographic screens in Sec.~\ref{sec:Application}.
We conclude in Sec.~\ref{sec:Discussion}, placing our result in context and commenting on connections to non-HRT extremal surfaces (e.g., ``minimax'' surfaces) and semiclassical quantum corrections.

\section{Background and definitions}
\label{sec:BackgroundDefs}

Let us begin by reviewing relevant definitions and background material.
We will largely follow the conventions and nomenclature from the literature on holographic screens, for example, as rendered in Ref.~\cite{Engelhardt:2020mme}.

\subsection{Topological definitions}
\label{ssec:TopologicalDefs}

We will mainly consider connected spacetimes $\mathcal{M}$ that possess a causal boundary $\mathcal{B}$ consisting of one or more timelike or null connected components.

Throughout this paper, we will be concerned with geometric objects constructed from light rays: causally defined regions, expansion of bundles of light rays, etc.
We begin, for completeness, by reviewing some terminology, following Refs.~\cite{Engelhardt:2018kcs,Nomura:2018aus,Bousso:2018fou}.
For a set of points $S\subset {\cal M}$, we define the chronological and causal future of $S$ as $I^+[S]$ and $J^+[S]$, denoting the sets of all points in ${\cal M}$ connected with any point in $S$ by timelike or null paths, respectively.
The chronological and causal pasts $I^-[S]$ and $J^-[S]$ are defined analogously.
We will assume that ${\cal M}$ is globally hyperbolic: it is free of closed causal curves and, for all points $p,q\in {\cal M}$, $J^+(p)\cap J^-(q)$ is compact (or, in the asymptotically-AdS case, when supplemented with  boundary conditions as discussed in Ref.~\cite{Avis:1977yn}).
We define the future (past) domain of dependence of $S$, $D^\pm[S]$, as the set of all $p\in {\cal M}$ for which all past (future) inextendible  causal curves through $p$ pass through $S$, and define $D[S]=D^+[S]\cup D^-[S]$.
We will refer to a surface as acausal if every pair of points on the surface is spacelike-separated; the weaker condition achronal is defined analogously, but allowing null-separated points.
A Cauchy slice $\Sigma$ is defined as an achronal codimension-one hypersurface for which $D[\Sigma]={\cal M}$.
We will at times refer to an achronal slice of the boundary ${\cal B} = \partial{\cal M}$ (i.e., ${\cal B}\cap\Sigma = \partial\Sigma$ for a Cauchy slice $\Sigma$) as itself a boundary $B$, trusting that the difference between $B$ and ${\cal B}$ will be clear in context.
From a closed, codimension-two, achronal surface $\nu \subset {\cal M}$, we can define four orthogonal null congruences---families of null geodesics hereafter referred to as light sheets---associated with the null rays launched either outward or inward, past- or future-directed, from $\nu$, spatially orthogonal to $\nu$ at the launch point.
Let us define the null vector on one side (e.g., ``outward'') as $k$ and on the other side (``inward'') as $\ell$, with $+k$ and $+\ell$ being future-directed, and write the four congruences as $N_{\pm k}[\nu]$ and $N_{\pm \ell}[\nu]$.
When convenient, we will also denote the union of future- and past-directed null congruences by $N_k[\nu] = N_{+k}[\nu] \cup N_{-k}[\nu]$ and $N_{\ell}[\nu] = N_{+\ell}[\nu] \cup N_{-\ell}[\nu]$.  
Geodesics can exit the light sheets if~\cite{Wald,Penrose} and only if~\cite{Akers:2017nrr} they encounter either a caustic or an intersection with a distinct null geodesic.
For $\nu$ Cauchy-splitting---that is, dividing a Cauchy surface $\Sigma \supset \nu$ into two sides $\Sigma^\pm$---we can write
\be
\begin{aligned}
N_{\pm k} [\nu] &=\dot I^{\pm}[\Sigma^\pm] - \Sigma^\pm \\
N_{\pm \ell}[\nu] &= \dot I^\pm[\Sigma^\mp]-\Sigma^\mp. 
\end{aligned}
\ee
where $\dot{}$ denotes the boundary of a set. See Fig.~\ref{fig:Defs}.
We will choose $\Sigma^+[\nu]$ to be on the side of $\nu$ toward which $k$ points and $\Sigma^-[\nu]$ to be on the side toward which $\ell$ points.
We define the outer wedge $O_W[\nu]=D[\Sigma^+[\nu]]$ and inner wedge $I_W[\nu]=D[\Sigma^-[\nu]]$.

\begin{figure}[h]
	\centering
	\includegraphics[width=0.9\linewidth]{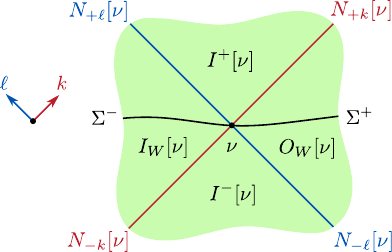}
	\caption{The causal wedges and null light sheets defined by a Cauchy hypersurface $\Sigma$, a Cauchy-splitting surface $\nu$, and the two null directions $k$ and $\ell$.}
	\label{fig:Defs}
\end{figure}

\subsection{Dynamical definitions}
\label{ssec:DynamicalDefs}

Having concluded the topological preliminaries, we now turn to dynamics.
We will define the 2+2 formalism for general relativity, in which the Einstein equations governing the evolution of spacetime can be recast in terms of these light-centric objects.
This formulation is known as the {\it characteristic initial data formalism} for general relativity~\cite{Rendall:2000,Brady:1995na,ChoquetBruhat:2010ih,Luk:2011vf,Chrusciel:2012ap,Chrusciel:2012xf,Chrusciel:2014lha}, and one can view it as akin to the more familiar 3+1 Arnowitt-Deser-Misner formalism, but with spatial Cauchy data replaced with Cauchy data on a null surface.
Given a Cauchy-splitting surface $\nu$, let us focus on the light sheet along the $k$-direction, $N_k[\nu]$.
We will define $k$ to be affinely parametrized as it is parallel transported along itself, and $\ell$ to be parallel transported along $k$ but continually rescaled so that $k\cdot \ell=-1$.
We define the induced metric as $q_{\mu\nu} = g_{\mu\nu} + k_\mu \ell_\nu +k_\nu \ell_\mu$, in terms of which we have the null extrinsic curvature,
\be 
(B_k)_{\mu\nu} = q_\mu^{\;\;\rho}q_\nu^{\;\;\sigma} \nabla_\sigma k_\rho.
\ee
This fundamental object defines the expansion,
\be 
\theta_k = q^{\mu\nu}(B_k)_{\mu\nu},
\ee
and shear,
\be
(\varsigma_k)_{\mu\nu} = (B_k)_{(\mu\nu)}  - \tfrac{1}{D-2}\theta_k q_{\mu\nu}.
\ee
We can also define the twist one-form, describing frame-dragging,
\be
\omega_\mu = -\ell_\nu q_\mu^{\;\;\rho} \nabla_\rho k^\nu.
\ee
Finally, we will write null-contracted indices as $(\cdot)_k = (\cdot)_\mu k^\mu$, and we define a Lie derivative ${\cal L}_k$ along $k$, a covariant derivative along a constant-affine slice ${\cal D}_\mu = q_\mu^{\;\;\nu} \nabla_\nu$, and the intrinsic Ricci curvature ${\cal R}$ of a slice.
In terms of these fundamental geometric data,  the dynamics of the light sheet dictated by general relativity become a set of three evolution equations, known collectively as the constraint equations,
\be
\begin{aligned}
\nabla_k \theta_k &= -\frac{1}{D-2}\theta_k^2 - \varsigma_k^2 - 8\pi GT_{kk} \\
q_\mu^{\;\;\nu} {\cal L}_k \omega_\nu &= -\theta_k \omega_\mu {+}\tfrac{D-3}{D-2} {\cal D}_\mu \theta_k  {-} ({\cal D}\cdot \varsigma_k)_\mu \,{+} \,8\pi G T_{\mu k}\\
\nabla_k \theta_\ell &= -\frac{1}{2}{\cal R} -\theta_k \theta_\ell +\omega^2  + {\cal D}\cdot \omega + 8\pi G T_{k\ell},
\end{aligned} \label{eq:CIDF}
\ee
where the energy-momentum tensor can include a cosmological constant.
The first equation in Eq.~\eqref{eq:CIDF} is the Raychaudhuri equation, the second is the Damour-Navier-Stokes equation, and the third is the cross-focusing equation.
Along $N_\ell[\nu]$, where $\ell$ is affinely parameterized and $k$ is parallel transported and rescaled such that $k\cdot \ell=-1$, the same equations apply, with $k\leftrightarrow \ell$ and $\omega\rightarrow -\omega$.
The principal result of the characteristic initial data formalism is that, given a set of null congruences $\Sigma$ on which data has been specified satisfying the constraint equations~\eqref{eq:CIDF}, there exists a unique spacetime satisfying the Einstein equations in $D[\Sigma]$.

\subsection{Screen-related definitions}
\label{ssec:ScreenDefs}

We now turn to a few other definitions and results that are more closely related to horizons and holography.
First, concentrating on the $k$-directed light sheet, we describe a surface as normal if $\theta_k >0$ and $\theta_\ell<0$, \mbox{(anti-)}trapped if both $\theta_k$ and $\theta_\ell<0$ ($>0$), and marginally (anti-)trapped if $\theta_k = 0$ and $\theta_\ell \leq 0$ ($\geq 0$).
An extremal surface satisfies $\theta_k = \theta_\ell =0$.
For a marginal surface $\sigma$, we define strict spacetime stability as the requirement that $\nabla_k \theta_\ell \leq 0$ everywhere on $\sigma$, with equality only if $\sigma$ is extremal; as shorthand, we refer to a marginal surface obeying strict spacetime stability simply as {\it stable}.

It will also be useful to specify a refinement of the notion of a marginal surface: a minimal marginal surface $\mu$ (that is, a {\it minimar} surface~\cite{Engelhardt:2018kcs}) is a stable marginal surface homologous to the boundary $B$ and for which there exists a Cauchy slice $\Sigma$ of $O_W[\mu]$ on which, among all $\nu \subset \Sigma$ homologous to $B$, $\mu$ has minimal area, $A[\mu]\leq A[\nu]$. 

We define a {\it Hubeny-Rangamani-Takayanagi (HRT) surface} as an extremal surface $X$ homologous to the boundary $B$, such that there exists a Cauchy slice $\Sigma$ of $D[O_W[X]]$ on which $X$ is minimal.
As such surfaces can be identified via a maximin optimization~\cite{Wall:2012uf}, we will sometimes refer to HRT surfaces as ``being maximin.''
Though HRT surfaces can be defined for subregions of the boundary, throughout, we will be interested in taking $B$ to be a complete connected component.
One can show that HRT surfaces are minimar~\cite{Engelhardt:2018kcs}.
HRT surfaces are of particular interest to holography, since the area of the HRT surface, divided by $4G\hbar$, corresponds to the fine grained (von~Neumann) entropy associated with the entanglement of the state on $B$ with its complement, 
\be
S(\rho_B) = -{\rm tr}(\rho_B \log \rho_B ) = \frac{A[X]}{4G\hbar}.
\ee 
This equality has been proven in the context of AdS/CFT, and it provides the appropriate generalization of the RT prescription for computing the entropy in terms of minimal area surfaces, which holds for static spacetimes.
We remark that, while we certainly have holographic applications in mind, we do not specialize to asymptotically AdS spacetimes when defining an HRT surface.
 
Finally, we define a {\it holographic screen}, a special surface with a remarkable area theorem~\cite{Bousso:2015mqa,Bousso:2015qqa} and compelling links to the coarse graining of holographic information~\cite{Engelhardt:2017aux,Engelhardt:2018kcs}.
A {\it future holographic screen} is a codimension-one hypersurface for which there exists a slicing into compact, acausal codimension-two surfaces, each of which is strictly marginally trapped, by which we mean $\theta_k=0$ and $\theta_\ell <0$.
Such a slicing is called a foliation, and the marginally trapped slices are termed leaves.
A past holographic screen is defined analogously, for marginally antitrapped surfaces.
Throughout, we take our holographic screen to have $C^2$ smoothness.
A holographic screen generalizes the more familiar notion of an apparent horizon (e.g., for a black hole), which is a hypersurface foliated by the outermost marginally trapped leaves on some Cauchy slice.
Like an apparent horizon, a holographic screen depends on the slicing (i.e., on the coordinate choice or observer).
Nonetheless, the holographic screen carries physical significance, in the form of an area theorem.
Defining a scalar field $\tau$ on the future holographic screen $H$, so that each leaf $\sigma$ is a surface of constant $\tau$, one can define the tangent vector $h^\mu = (\partial_\tau)^\mu = \alpha \ell^\mu + \beta k^\mu$.
Throughout, we assume Einstein's equations plus the null energy condition (NEC), $T_{uu} \geq 0$ for any null vector $u$. As shown in Ref.~\cite{Bousso:2015qqa}, it is then useful to invoke a handful of assumptions defining a {\it regular} holographic screen, specifically, the genericity conditions that
\begin{enumerate}
\item $R_{kk} + \varsigma_k^2 > 0$ everywhere on $H$, and
\item $\alpha = 0$ only in a measure-zero region of $H$ forming the boundary between regions where $\alpha \gtrless 0$,
\end{enumerate}
as well as the technical requirements that 
\begin{enumerate}
\setcounter{enumi}{2}
\item dividing $H$ up into its regions of $\alpha$ with definite sign, each region contains either a complete leaf or is timelike; and
\item for each leaf $\sigma$, there exists a Cauchy slice for which $\sigma$ is Cauchy-splitting.
\end{enumerate}
From these definitions, it follows~\cite{Bousso:2015qqa} that $\alpha<0$ throughout $H$.

As a consequence, only certain transitions in the signature of $h$ are permitted, precisely such that the flow along the screen is either past- or outward-directed.
The future holographic screen can then be shown to satisfy an area law: $A[\sigma(\tau_2)] > A[\sigma(\tau_1)]$ for all $\tau_2>\tau_1$.
This area law is distinct from Hawking's area theorem, which pertains to event horizons rather than apparent horizons, though the two have been unified into a larger class of objects~\cite{Nomura:2018aus}.

Hawking's area theorem and its recasting as the second law of thermodynamics applied to the Bekenstein-Hawking entropy of black holes suggests the existence of an interpretation of holographic screens in terms of a coarse graining of some holographic entropy.
Concretely, an information-theoretic interpretation for the outermost spacelike portion of the holographic screen---that is, the apparent horizon---was found in Ref.~\cite{Engelhardt:2017aux}.
Given a holographic asymptotically-AdS spacetime, one defines the {\it outer entropy} $S^{\mrm{(outer)}}[\sigma]$ associated with a marginally trapped surface $\sigma$ as the largest von Neumann entropy among those of boundary states $\rho_B$ on $B$ such that $\rho_B = \tr_{\bar B}\rho$ and the geometry $\mathcal{M}_\rho$ dual to $\rho$ contains $O_W[\sigma]$:
\begin{equation} \label{eq:Sout}
S^{\text{(outer)}}[\sigma] = \max_{\rho \, : \, O_W[\sigma] \subset \mathcal{M}_\rho} S(\rho_B).
\end{equation}
Equivalently, via the HRT formula, $S^\mrm{(outer)}[\sigma]$ is equal to the maximal area, in units of $4G\hbar$, of any HRT surface constructible in the spacetime subject only to the NEC and to holding the outer wedge $O_W[\sigma]$ fixed.\footnote{Outer entropy is a quantity that is defined intrinsically in terms of the bulk spacetime. The corresponding boundary-centric quantity is simple entropy. We will not make use of it here, but details are explained in Refs.~\cite{Engelhardt:2017aux,Engelhardt:2018kcs}.}
The outer entropy is a coarse grained holographic entropy, as it describes a von~Neumann entropy maximized subject to only partial information about the spacetime, namely $O_W[\sigma]$.

For $\mu$ minimar, Refs.~\cite{Engelhardt:2017aux,Engelhardt:2018kcs} showed that
\begin{equation} \label{eq:Sout_area}
S^{\mrm{(outer)}}[\mu]= \frac{A[\mu]}{4G\hbar}.
\end{equation}
The fact that outer wedges nest along the apparent horizon, i.e., $O_W[\mu(\tau_2)]\subset O_W[\mu(\tau_1)]$ for $\tau_2 > \tau_1$, then explains the growth of the area in entropic terms, as the outer entropy is computed with more and more data about the spacetime integrated out (see, e.g., Fig.~2(b) of Ref.~\cite{Engelhardt:2017aux}).
A pressing open question, of particular relevance to the case of collapsing black holes and cosmological applications of holographic screens, is to provide an analogous, explicit holographic entropy construction for timelike portions of the screen.
In Sec.~\ref{sec:Application}, we comment on challenges inherent to doing so.
However, toward that end, understanding and characterizing where HRT surfaces can be located, in generality, relative to a holographic screen is a question of critical importance to any future attempt at providing such an entropic formulation.

\section{Sequestering holographic screens with extremal surfaces}
\label{sec:Sequestering}

Suppose we have a spacetime $\mathcal{M}$ with boundary $\mathcal{B}$ and that the NEC is satisfied.
Given an HRT surface $X$ in $\mathcal{M}$, we want to know where holographic screens $H$ can be located relative to $X$.
Without loss of generality, let us consider the case of future holographic screens, which are foliated by marginally trapped surfaces $\sigma$ as defined in Sec.~\ref{sec:BackgroundDefs}.
To adapt our arguments to past holographic screens, simply take the time-reverse and replace marginally trapped surfaces with marginally antitrapped ones. 
Then, three properties of holographic screens and HRT surfaces provide rules that we can use to constrain where $\sigma$ can be situated with respect to $X$, and thus determine where $H$ can be located relative to $X$. 
These rules can be summarized as follows:
\begin{enumerate}[label={(\roman*)}]
	\item Cross-sectional areas of the null congruence $N_{k}[\sigma]$ are nondecreasing toward the marginal surface $\sigma$; similarly, cross-sectional areas of $N_{k}[X]$ and $N_{\ell}[X]$ are nondecreasing toward the HRT surface $X$.\label{rule1}
	\item Areas of leaves foliating $H$ strictly increase toward the past and exterior along $H$.\label{rule2}
	\item There exists a Cauchy slice $\Sigma$ containing the HRT surface $X$, for which any Cauchy-splitting surface in $\Sigma$ has area at least as large as $X$. \label{rule3}
\end{enumerate}

Let us take a closer look at where each of these rules come from. 
Rule~\ref{rule1} stems directly from the definition of marginality, the Raychaudhuri equation~\eqref{eq:CIDF}, and the NEC.
Marginally trapped surfaces $\sigma$ that foliate the future holographic screen have vanishing null expansion in the marginal direction $k$, i.e., $\theta_k = 0$ for any $\sigma \in H$. 
Likewise, an HRT surface $X$ is extremal, and therefore marginal in both the $k$ and $\ell$ directions, so $\theta_k = 0$ and $\theta_\ell = 0$ at every point on $X$ as well. 
The Raychaudhuri equation~\eqref{eq:CIDF} then implies that $\theta_{\pm k} \leq 0$ along the null congruences $N_{\pm k}[\sigma]$ emanating from $\sigma$, and similarly $\theta_{\pm k} \leq 0$ and $\theta_{\pm \ell} \leq 0$ along $N_{\pm k}[X]$ and $N_{\pm \ell}[X]$. 
In other words, if $\nu$ is a section of the null congruence $N_{k}[\sigma]$, then we have that $A[\nu] \leq A[\sigma]$.
Likewise, $A[\nu] \leq A[X]$ holds for $\nu \in N_{k}[X], N_{\ell}[X]$. 
Rule~\ref{rule2} is the area law for holographic screens~\cite{Bousso:2015mqa, Bousso:2015qqa} reviewed in Sec.~\ref{ssec:ScreenDefs}. 
Finally, rule~\ref{rule3} comes from the definition of an HRT surface as introduced in Sec.~\ref{ssec:DynamicalDefs}. 

\begin{figure}[h]
	\centering
	\includegraphics[width=0.85\linewidth]{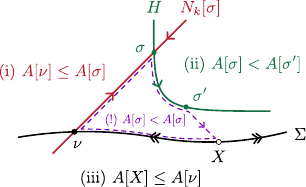}
	\caption{A configuration of $H$ and $X$ showing how rules~\ref{rule1}-\ref{rule3} apply in practice. Arrows on hypersurfaces represent the direction of increasing cross-sectional area. The double arrow on $\Sigma$ indicates that a notion of increasing area is only true when comparing a Cauchy-splitting surface $\nu \in \Sigma$ with the HRT surface $X$. The purple dashed line highlights the fact that in this configuration a closed, directed path can be formed that is inconsistent with the rules. This is an example of what we call a \emph{forbidden loop}.}
	\label{fig:ThreeRules}
\end{figure}

The illustration in Fig.~\ref{fig:ThreeRules} demonstrates all three rules in action for one configuration of a screen $H$ in the presence of an HRT surface $X$.
In this diagram, the surface $\nu$ is defined by the intersection\footnote{In the context of maximin surfaces, this is also known as the representative of a surface $\nu$ on a Cauchy slice. See Refs.~\cite{Engelhardt:2018kcs,Wall:2012uf}, for example.} of the null congruence $N_{k}[\sigma]$ and the Cauchy slice $\Sigma$, i.e., $\nu = N_{k}[\sigma] \cap \Sigma$.
Note that, in this particular case, a closed, directed path in the Penrose diagram along which the area of foliating surfaces increases can be formed from $N_{-k}[\sigma]$, $H$, $N_{+\ell}[X]$, and $\Sigma$ by invoking the above three rules. 
Chaining the inequalities along these hypersurfaces and comparing cross-sectional areas along this closed path leads to a contradiction: $A[\sigma] < A[\sigma]$.
Note that here we are implicitly assuming that $N_{+\ell}[X] \cap H = \sigma'$ is a leaf of $H$ in order to invoke rule~\ref{rule2}. This will not generically be the case in the absence of symmetry; nevertheless, the basic idea remains unchanged. In the interest of clarity, we will continue with this assumption for now and later relax it in \Sec{ssec:Nonspherical}.

We thus find that Fig.~\ref{fig:ThreeRules} in fact illustrates an example where a configuration of $H$ and $X$ is inconsistent with the above rules. 
This observation gives us the simple criterion: if such a closed, directed path (i.e., hypersurface) can be formed with $\sigma \in H$ and $X$ being any two nodes (i.e., inscribed surfaces) in the path, then such a configuration of $H$ and $X$ is forbidden as it results in a contradictory area inequality. 
We will henceforth refer to these closed paths as \emph{forbidden loops}. 
Any configuration of $H$ and $X$ that does not form a forbidden loop is then consistent with the three rules and is allowed. 
It is important to point out that the crux of this argument relies on the fact that rule~\ref{rule2}, the area law for holographic screens, is a strict inequality. 
If all inequalities in the chain were weak, then a contradiction could not be established (albeit with the loophole being a very finely tuned spacetime in which the cross-sectional area is constant along the closed path).

An HRT surface $X$ naturally partitions the spacetime into four regions through its orthogonal null congruences $N_{k}[X]$ and $N_{\ell}[X]$: the causal future $I^+[X]$, outer wedge $O_W[X]$, inner wedge $I_W[X]$, and causal past $I^-[X]$. 
By applying the ``forbidden loops'' argument of Fig.~\ref{fig:ThreeRules} to each of these four wedges, we can explicitly determine where a screen leaf $\sigma$ can be located relative to $X$, and thus what types of structures of $H$ are permitted relative to this partition. 
A crucial observation can immediately be made: like the scenario in Fig.~\ref{fig:ThreeRules}, a forbidden loop can be found whenever a segment of a holographic screen passes through any of the null congruences $N_{k}[X]$ and $N_{\ell}[X]$ fired from $X$.
These forbidden loops are explicitly illustrated in Fig.~\ref{fig:ForbiddenLoops}.
Interestingly, note in particular that the forbidden loop for a spacelike segment crossing $N_{-\ell}[X]$ does not rely on the Cauchy slice $\Sigma$, and is thus agnostic of the maximin nature of $X$. 
We will return to this curiosity when examining edge cases in \Sec{ssec:EdgeCases} and find that we do not need the forbidden loop argument to prove that screens cannot cross $N_{-\ell}[X]$.
By virtue of the fact that forbidden loops appear whenever holographic screens cross the null congruences of $X$, we arrive at the conclusion that holographic screens must remain \emph{sequestered} within a given wedge (at least for the symmetric case as noted above). We will prove the  same statement, without assuming symmetry of the spacetime, in \Sec{ssec:Nonspherical}.

	\begin{figure*}[t]	
		\centering
		\includegraphics[width=0.78\textwidth]{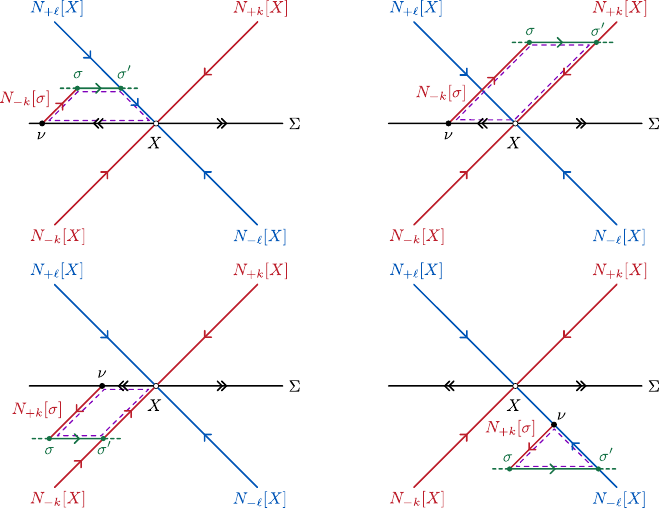}
		\caption{Forbidden loops (purple dashed lines) for spacelike segments of a holographic screen (green solid line) passing through each of the four null congruences $N_{\pm k}[X]$ and $N_{\pm \ell}[X]$ fired from $X$ (red and blue solid lines respectively). In each case, in addition to the screen segment and one of the four null congruences, the forbidden loop is closed by $N_{\pm k}[\sigma]$, a null congruence of $\sigma$ in the marginal null direction $k$, and possibly the Cauchy slice $\Sigma$. Analogous results hold for timelike or mixed signature holographic screen segments.}
		\label{fig:ForbiddenLoops}
	\end{figure*}

\subsection{Allowed screen trajectories}
\label{ssec:AllowedScreens}
\vspace{-1mm}
Having established that holographic screens must be sequestered within each of the four wedges defined by $X$, we can further determine how screens are allowed to behave within each wedge. 
To study this question, we will systematically examine the trajectory of screens that are timelike or spacelike early on in their flow. By ``early'' and ``late,'' we  mean in the sense of flow along the screen, i.e., as $\tau\rightarrow -\infty$ or $+\infty$, respectively.
Likewise, when we say that a region of a screen is ``timelike'' or ``spacelike,'' we mean that $h_\mu h^\mu < 0$ or $h_\mu h^\mu > 0$ on that region, respectively.
A screen that is timelike early on in its flow has leaves on which $h_\mu h^\mu < 0$ everywhere as $\tau \rightarrow -\infty$, with other combinations defined analogously.
(Recall, however, that in the absence of sufficient symmetry $h^\mu$ need not have uniform character on a single leaf.)
The allowed trajectories of future holographic screens are summarized in Fig.~\ref{fig:AllowedScreens}.
\vspace{-1.5mm}

\subsubsection{Causal future $I^+[X]$}
\vspace{-1.5mm}
Let us first focus our attention on screens with early timelike flows. 
In the causal future of $X$, these are screens with a timelike segment close to future infinity. 
By the area law for holographic screens~\cite{Bousso:2015mqa, Bousso:2015qqa}, timelike segments of future holographic screens may either remain timelike and past-directed, or turn outward into a spacelike segment; turning spacelike inward violates monotonicity of its cross-sectional area. 
Since we already know that holographic screens cannot cross $N_{k}[X]$ and $N_{\ell}[X]$, if a timelike segment coming in from future infinity is to remain timelike, it can only end on or go through $X$.
As shown in Fig.~\ref{fig:HRT_EndLoops}, however, it is easy to see that a forbidden loop arises in this case as well, so the only allowed screens with early timelike flows are those that have late spacelike flows that asymptote to $N_{+k}[X]$. 
On the other hand, by sequestration, screens with early spacelike flows in the future of $X$ must asymptote, backward along the flow, to the null congruence $N_{+\ell}[X]$.
In a similar fashion, spacelike segments of future holographic screens may either remain spacelike toward the exterior, or turn timelike toward the past, as turning timelike toward the future violates the area law. 
Since such screens must again not cross $N_{k}[X]$ or $N_{\ell}[X]$, or pass through $X$, they must also have late spacelike flows asymptoting to $N_{+k}[X]$ just as for the timelike case.
Note that while consistency with the rules above constrain the early and late flow behavior of $H$ to the future of $X$, the screen itself is allowed to have arbitrarily many intermediate transitions between spacelike and timelike regions of $H$, as long as monotonicity as stipulated by the area law is respected.

\begin{figure}[t]
	\centering
	\includegraphics[width=0.95\linewidth]{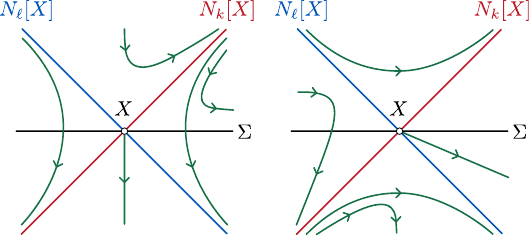}
	\caption{Representative trajectories of future holographic screens that are consistent with the rules outlined above. The left figure shows all allowed trajectories with early timelike flows while the right figure shows all allowed trajectories with early spacelike flows. In both cases, sequestration stipulates that no screens are allowed to cross the null congruences $N_{\pm k}[X]$ (red line) and $N_{\pm \ell}[X]$ (blue line) fired from the HRT surface $X$. Note that although screens may start on $X$ in some cases, they can never end on $X$. Here, we are depicting the allowed early and late asymptotic behavior along the screen; any number of intermediate transitions between spacelike and timelike regions is allowed.}
	\label{fig:AllowedScreens}
\end{figure}

\subsubsection{Outer wedge $O_W[X]$}

In the outer wedge, by sequestration, screens that have early timelike flows must asymptote to $N_{+k}[X]$ toward the future. 
As we evolve along the flow, the screen can continue in a timelike fashion and asymptote to $N_{-\ell}[X]$ toward the past so as to not cross null congruences of $X$, or turn spacelike toward the outer boundary $\mathcal{B}$ in accordance with the area law. 
There is no penalty for the screen to cross the Cauchy surface $\Sigma$ defining $X$ as this does not generate forbidden loops, and so late spacelike flows can be either to the future or past of $\Sigma$. 
On the other hand, screens with early spacelike flows must begin at the HRT surface $X$ since sequestration dictates that they cannot come from another wedge. 
These early spacelike segments can have the same late flow behavior as early timelike segments, i.e., they can have late spacelike flows to the future or past of $\Sigma$, or turn timelike to the past and asymptote to $N_{-\ell}[X]$. 
Note that screens in the outer wedge all have to end at the boundary $\mathcal{B}$.

\begin{figure}[h]
	\centering
	\includegraphics[width=\linewidth]{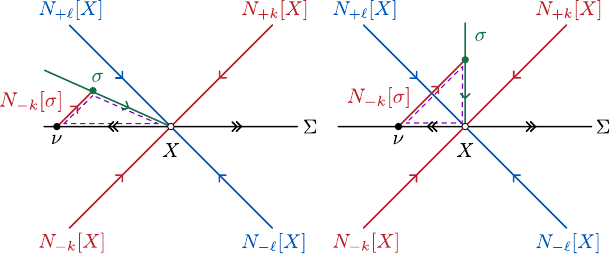}
	\caption{Forbidden loops (purple dotted lines) that appear when a holographic screen (green line) in the interior $I_W[X]$ (left panel) or future $I^+[X]$ (right panel) ends at the HRT surface $X$. Analogous loops can be found when the screen is to the past of the Cauchy slice $\Sigma$ in the interior.}
	\label{fig:HRT_EndLoops}
\end{figure}

\vspace{-3mm}

\subsubsection{Inner wedge $I_W[X]$}

In the inner wedge, screens with early timelike flows must asymptote to $N_{+\ell}[X]$ by sequestration. 
These screens can again turn spacelike outward, but like in the future wedge, will inevitably run into $X$.
Again, this case can be ruled out by the presence of forbidden loops as shown in Fig.~\ref{fig:HRT_EndLoops}, so screens with early timelike flows can only have late timelike flows that asymptote to $N_{-k}[X]$ by sequestration. 
Screens with early spacelike flows can have spacelike segments either to the future or past of the Cauchy surface $\Sigma$ like screens in the outer wedge, but must turn toward the past in a timelike fashion by virtue of the area law and again asymptote to $N_{-k}[X]$ for late flows. 

\vspace{-3mm}

\subsubsection{Causal past $I^-[X]$}

Similar to screens with early spacelike flows in the outer wedge, screens with early timelike flows in the causal past of $X$ can only start at the HRT surface $X$ due to sequestration. 
From here, the screen can have late timelike flows toward past infinity, or turn spacelike outward and asymptote to $N_{-\ell}[X]$ for late spacelike flows. 
Screens with early spacelike flows in the past of $X$ must asymptote, backward along the flow, to $N_{-k}[X]$ by sequestration, then continue to have a late timelike flow to past infinity following the area law, or asymptote to $N_{-\ell}[X]$ with a late spacelike flow by sequestration.

\subsection{Nonsymmetric spacetimes}
\label{ssec:Nonspherical}

In using forbidden loops to argue for the sequestration of holographic screens, we made use of the fact that the area law for holographic screens is strictly monotonic.
It was then straightforward to show that a contradiction arises when we compare areas along $H$ in these forbidden loops. 
This comparison of cross-sectional areas is only meaningful, however, when $H \cap N_{\pm u}[X] = \sigma'$ for $u = k,\ell$---for whichever choice the intersection is nonempty---is itself a leaf of $H$. 
While this is true in symmetric cases, e.g., spacetimes with spherical symmetry, the intersection of $H$ with a null congruence of $X$ will not generically be a leaf of the screen.
Nevertheless, we can refine our argument and show that holographic screen sequestration continues to hold for a large class of spacetimes.

As before, let $\sigma(\tau)$ denote leaves foliating a holographic screen $H$ with $\tau$ parametrizing the flow along the screen so that $A[\sigma(\tau_1)] < A[\sigma(\tau_2)]$ for $\tau_1 < \tau_2$.
We will now consider the case in which the null congruence $N_{\pm u}[X]$ intersects $H$ in an oblique manner such that the intersection $H \cap N_{\pm u}[X] = Y$ is some compact spacelike surface that is not a leaf, i.e., not a marginal surface.
More generally, since $k, \ell$ are defined everywhere by the outgoing null congruences shot out from the leaves $\sigma(\tau)$ of $H$, and parallel transporting along themselves, the null congruences fired from $X$ in nonsymmetric spacetimes will no longer be generated by $k, \ell$ as they are not necessarily orthogonal to $X$.\footnote{See the discussion in Step~2, Sec.~5.2 of Ref.~\cite{Engelhardt:2018kcs}, for example.} 
With this caveat in mind, we will call the $X$-orthogonal null congruences $N_{\pm u}[X]$ for $u = \tilde{k}, \tilde{\ell}$, which are in general nonparallel to the $\sigma(\tau)$-orthogonal $k, \ell$ unless the spacetime exhibits sufficient (e.g., spherical) symmetry.
The geometry of these surfaces is illustrated in Fig.~\ref{fig:DeformProof} for the case of $Y = H \cap N_{\tilde{k}}[X]$ without loss of generality (a similar cartoon can be drawn for $Y = H \cap N_{\tilde{\ell}}[X]$).

\begin{figure}[h]
	\centering
	\includegraphics[width=0.8\linewidth]{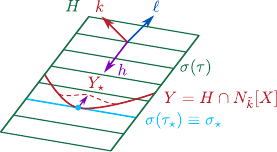}
	\caption{A holographic screen $H$ with transverse spacelike directions partially restored in the direction perpendicular to the flow generated by the tangent vector $h = \alpha \ell + \beta k$. The canonical foliation by marginal surfaces $\sigma(\tau)$ is given by constant-$\tau$ slices along $H$. $Y$ is an oblique intersection of $H$ with the null congruence $N_{\tilde{k}}[X]$. Infinitesimally deforming the intersection $Y$ in the $-h$ direction locally at a point (light blue dot) that shares a normal vector to a complete leaf $\sigma_{\star}$ gives another surface $Y_{\star} \in H$ (dashed red line) with strictly smaller area than $Y$.}
	\label{fig:DeformProof}
\end{figure}

Since $Y$ is not a leaf of $H$, we can no longer use the monotonic area law and conclude that $A[\sigma(\tau)] < A[Y]$ for some leaf along the flow before $H$ meets $N_{\pm u}[X]$.
To employ the same arguments as before, we would like for there to be some compact spacelike surface $Y_{\star}$ forming a different slice of $H$ such that $A[Y_{\star}] < A[Y]$ and such that the map along integral curves of $h$ from $Y$ to $Y_\star$ is in the $-h$ direction.
We stress that we do not need an area law along the screen for nonmarginal foliations analogous to the area law of Refs.~\cite{Bousso:2015mqa, Bousso:2015qqa}.
We simply require the existence of a lesser-area surface $Y_{\star}$.
If such a surface exists, then we can tweak the forbidden loop argument to use $Y$ and $Y_{\star}$ in lieu of marginal surfaces $\sigma(\tau)$ foliating the screen.
Once the additional ingredients are in place, we can once again chain area inequalities, for any of the forbidden loops used in our previous arguments, by the following steps: 
\begin{enumerate}[label={(\alph*)}]
	\item By assumption, we have that $A[Y_{\star}] < A[Y]$.\label{step1}
	\item Then $A[Y] \leq A[X]$ from focusing along $N_{\pm u}[X]$.\label{step2}
	\item We also have $A[X] \leq A[N_{\pm k}[Y_{\star}] \cap \Sigma]$ from the definition of an HRT surface. \label{step3}
	\item Lastly, $A[N_{\pm k}[Y_{\star}] \cap \Sigma] \leq A[Y_{\star}]$ by focusing along $N_{\pm k}[Y_{\star}]$. \label{step4}
	\item Chaining everything together, we arrive at $A[Y_{\star}] < A[Y_{\star}]$, a contradiction.
\label{step5}
\end{enumerate}

A few clarifying comments are in order.
First, recall again that the vector field $k$ is defined everywhere by launching outgoing, future-directed null congruences from each of the leaves of $H$ and parallel transporting $k$ along itself.
In particular, for nonsymmetric spacetimes, the notation $N_{\pm k}[Y_\star]$ does {\it not} refer to an orthogonal congruence launched from $Y_\star$, because such a vector field will not in general even be parallel to $k$~\cite{Engelhardt:2018kcs}.
So by $N_{\pm k}[Y_\star]$, we mean the union of $N_{\pm k}[p]$ for each $p\in Y_\star$, where $N_{\pm k}[p]$ is defined by the pencil of affine generators launched from $p$ within the orthogonal null congruence $N_{\pm k}[\sigma_p]$ for $\sigma_p$, the leaf of $H$ within which $p$ resides.
For step~\ref{step4} above, it is important to note that any area difference between $Y_\star$ and complete slices of $N_{\pm k}[Y_\star]$ is accounted for entirely by the null expansion along $k$. 
In other words, there is no area difference associated with the tilt of $Y_\star$ relative to constant-affine surfaces, precisely because the congruence is null.
That is, consider any null congruence $N_{\pm u}[\nu]$, launched from any surface $\nu$, and choose two complete slices $\rho_1$ and $\rho_2$ within the congruence such that $\rho_1$ is to the past of $\rho_2$ and $\theta_u$ (defined with respect to $N_{\pm u}[\nu]$) vanishes everywhere on $\rho_2$; then $A[\rho_1]\leq A[\rho_2]$ by focusing, regardless of whether $\rho_{1,2}$ are constant-affine surfaces.

All figures containing forbidden loops depicted hitherto in Sec.~\ref{sec:Sequestering} then carry over to nonspherical spacetimes with the replacement $\sigma' \rightarrow Y$ and $\sigma \rightarrow Y_{\star}$.\footnote{In this case, however, the diagrams exaggerate the separation between $Y_{\star}$ and $Y$ on $H$, since we will show in Sec.~\ref{sssec:Ystar} that we can always find a $Y_{\star}$ infinitesimally close to $Y$.}
For all of our previous sequestration results to carry over to nonsymmetric spacetimes, all that remains is to show that such a surface $Y_{\star}$ exists. 
 
We will do this by first showing that there exists at least one point $x$ on both the intersection $Y$ and a leaf $\sigma_{\star}$ that share a common tangent space $T_x Y = T_x \sigma_{\star}$ within $H$, or equivalently, for which their normal vectors within $H$ at $x$ are parallel.
Next, we will demonstrate an area law for infinitesimal elements of area $\delta A(p)$ for any point $p \in \sigma(\tau)$ on marginal leaves foliating the screen.
Finally, we prove using this infinitesimal area law that a first-order local deformation around this point $x$ is sufficient to guarantee that a surface $Y_{\star}$ exists such that $A[Y_{\star}] < A[Y]$, allowing us to complete the chain of inequalities.
While finite deformations of $Y$ in the $-h$ direction might take us off of the screen $H$, infinitesimal deformations, by definition, remain within $H$ to linear order.
The geometric content of this proof is explicitly shown in Fig.~\ref{fig:DeformProof}.
	
\subsubsection{Existence of points on $Y$ and $\sigma(\tau)$ with parallel normals}

Let us first prove that in a large class of spacetimes there will always be at least one point on $Y$ and $\sigma_{\star}$ that share a common normal vector in $H$.
We will use the Poincar\'{e}-Hopf index theorem:

\begin{theorem}[Poincar\'{e}-Hopf \cite{H1886, Hopf1927}]\label{thm:Index}
	Let $v$ be a continuous vector field on a compact boundaryless manifold $Y$ with isolated zeroes at $x_i$, i.e., $v(x_i) = 0$. Then, the sum of the indices of the zeroes of $v$ is equal to the Euler characteristic of $Y$,
	\begin{equation} \label{eq:HairyBall}
	\sum_{i} \mathrm{Index}_{x_i} (v) = \chi(Y).
	\end{equation}
\end{theorem}

For $Y$ of dimension $n$, the index of $v$ at a zero located at $x_i \in Y$ is a nonzero integer defined as the degree of the map from the boundary of a closed ball enclosing $x_i$ to the $(n-1)$-sphere. 
In particular, this implies that any compact manifold with nonzero Euler characteristic must have at least one point where $v$ must vanish.
For compact manifolds with vanishing Euler characteristic, the Poincar\'{e}-Hopf theorem does not exclude the existence of zeroes, but rather it fails to guarantee their existence as indices of different zeroes can conspire to cancel out.
If a vector field is everywhere nonvanishing on a compact manifold however, the Poincar\'{e}-Hopf theorem immediately implies that the manifold has vanishing Euler characteristic.

We now establish the following:
\begin{lemma} \label{lem:Normals}
	Let $H$ be a future holographic screen and let $Y$ be a compact, boundaryless surface that is a cross section of $H$. 
	Then, for $\chi(Y) \neq 0$, there must be at least one leaf $\sigma_{\star}$ of $H$ such that at least one point exists on both $Y$ and $\sigma_{\star}$ where the normal vectors, within $H$, are parallel.
	For $\chi(Y) = 0$, the existence of such a point is still allowed, but not guaranteed.
\end{lemma}

\begin{proof}
	The following geometric construction is shown in Fig.~\ref{fig:HairyBall}.
	Since $H$ admits a foliation by marginally trapped surfaces, and since $Y$ is some cross section of $H$, any point $p \in Y$ is also contained in a unique leaf $\sigma_p$.\footnote{Any leaf can however contain more than one point of $Y$, i.e., $p, p' \in Y$ can share the same leaf, in which case $\sigma_p = \sigma_{p'}$.}
	Let the normal vectors within $H$ at any point $p\in Y$ and $\sigma_p$ be denoted by $n_Y(p)$ and $n_{\sigma_p}(p)$ respectively.\footnote{By definition, $n_{\sigma_p}\propto h$. The fact that $Y$ is codimension-one with respect to $H$ ensures that there is a unique normal vector $n_Y (p)$ at any point $p \in Y$. Were $Y$ to be of any higher codimension, there would be multiple normal vectors to $Y$ at $p$.}
	The normal vectors $n_Y(p)$ and $n_{\sigma_p}(p)$ will generally not be parallel to each other for any given point $p$.
	Now, define a vector field $v(p)$ at each point $p$ in $Y$ as the projection of the normal vector $n_{\sigma_p}(p)$ onto $T_p Y$, the tangent space of $Y$ at $p$, in other words, let $v(p) = P_{T_p Y}  \left( n_{\sigma_p}(p) \right)$, where $P_{T_p Y}$ is the projector onto $T_p Y$.\footnote{We can construct this projector in the usual way. Writing the induced metric on $H$ as $h_{\mu\nu} = g_{\mu\nu} - h_\mu h_\nu/h^2$, we can define a further induced metric on $Y$ as $y_{\mu\nu} = h_{\mu\nu} - n_Y^\mu n_Y^\nu/n_Y^2$. Then the projection of a vector $w$ onto $Y$ is given by $y_{\;\;\nu}^{\mu} w^\nu$.}
	Then, by definition, $n_Y(p)$ and $n_{\sigma_p}(p)$ are parallel to each other when $v(p) = 0$ for some $p$.
	It immediately follows from Theorem~\ref{thm:Index}, the Poincar\'e-Hopf theorem, that if $\chi(Y) \neq 0$, then there must be at least one point on $Y$ where $v$ vanishes.
	Denoting these points by $x_i$, and the leaf containing them by $\sigma_{x_i}$ following Eq.~\eqref{eq:HairyBall}, $n_Y(x_i)$ and $n_{\sigma_{x_i}}(x_i)$ are therefore parallel at $x_i$.
\end{proof}

\begin{figure}[h]
	\centering
	\includegraphics[width=0.8\linewidth]{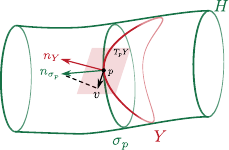}
	\caption{Geometric setup for the proof of Lemma~\ref{lem:Normals}.
		The normal vector within $H$ of $\sigma_p$ at $p$, denoted $n_{\sigma_p}$, is projected onto the tangent space $T_p Y$ of $Y$ at $p$ to define a vector field $v$ at every point $p \in Y$.
		When $\chi(Y) \neq 0$, the Poincar\'{e}-Hopf theorem guarantees the existence of at least one point where $v$ vanishes, i.e., where $n_Y$ and $n_{\sigma_p}$ are parallel to each other.}
	\label{fig:HairyBall}
\end{figure}

In what follows, it suffices to choose any one of these points if multiple exist, and label it $x$, along with the corresponding leaf $\sigma_{\star} \equiv \sigma(\tau_{\star})$.

	The special case where the intersection $Y$ is a topological two-sphere is precisely the statement of the hairy ball theorem.

\subsubsection{Local area law for infinitesimal elements of proper area}

Recall that $h^{\mu} = (d/d\tau)^{\mu}$ is the vector field tangent to the flow along a screen $H$ on any leaf $\sigma(\tau)$.
Corollary IV.4 of Ref.~\cite{Bousso:2015qqa} gives the local form of the rate of change of $A[\sigma(\tau)]$,
\begin{equation}
	\frac{d}{d \tau}A[\sigma(\tau)] = \int_{\sigma(\tau)} \sqrt{q^{\sigma(\tau)}} \alpha \theta_\ell^{\sigma(\tau)},
\end{equation}
where $q^{\sigma(\tau)}_{\mu \nu}$ is the induced metric on the leaf and $q^{\sigma(\tau)}$ is its determinant.

For any point $p \in \sigma(\tau)$ on the leaf, the infinitesimal element of proper area at $p$ is defined by
\begin{equation}
	 \delta A(p) = \sqrt{q^{\sigma(\tau)}(p)} ~dx^1 \cdots dx^n,
\end{equation}
for $p \equiv (x^1, \dots, x^n)$ in some coordinates when $\sigma(\tau)$ is of dimension $n$.
We now show that there is an analogous local area law for proper area elements $\delta A(p)$.
\begin{lemma}\label{lem:SmallAreaLaw}
	The rate of change per unit area of an infinitesimal element of proper area at a point $p \in \sigma(\tau)$ in the direction tangent to the flow along $H$ is
	\begin{equation} \label{eq:SmallAreaLaw}
		\frac{1}{\delta A(p)}\frac{d}{d \tau} \delta A(p) = \alpha (p) \theta_\ell^{\sigma(\tau)} (p).
	\end{equation}
\end{lemma}

\begin{proof}
	Recall that the tangent vector field orthogonal to each leaf can be written as the linear combination
	\begin{equation}
		h^\mu = \alpha \ell^\mu + \beta k^\mu,
	\end{equation}
	for some scalar fields $\alpha$ and $\beta$.
	By Theorem IV.2 of Ref.~\cite{Bousso:2015qqa}, $\alpha < 0$ everywhere on a future holographic screen $H$.
	We are still free to normalize the null generators $k^\mu$ and $\ell^\mu$ such that $k \cdot \ell = -1$ and $\ell^\mu$ is an affine generator.
	To preserve the cross-normalization however, rescaling $\ell^\mu$ also means rescaling $k^\mu$.
	Under this rescaling, $k^\mu$ will generically not be affinely parametrized, but will remain proportional to an affine generator, i.e.,
	\begin{equation}
			\ell^\mu = \left( \frac{d}{d \lambda_\ell} \right)^\mu, \qquad k^\mu = c_k \left( \frac{d}{d \lambda_k} \right)^\mu,
	\end{equation}
where $c_k$ is a scalar. 
	Since $\ell^\mu$ affinely parametrizes null congruences $N_{\ell}[\sigma]$ fired from leaves of the screen, the fractional rate of change of cross-sectional area of $N_{\ell}[\sigma]$ is given precisely by the null expansion
	\begin{equation}
		\theta_\ell = \frac{1}{\delta A} \frac{d}{d \lambda_\ell} \delta A.
	\end{equation}
	Then it follows that
	\begin{equation}
		\frac{1}{\delta A} \frac{d}{d \tau} \delta A = \frac{1}{\delta A} h^\mu \nabla_\mu \delta A = \alpha \theta_\ell + \beta c_k \theta_k.
	\end{equation}
	Equation~\eqref{eq:SmallAreaLaw} then follows from the fact that $\theta^{\sigma(\tau)}_k = 0$ everywhere on $\sigma(\tau)$.
	The monotonicity of this infinitesimal area law follows from the property that $\alpha(p) < 0$ and $\theta^{\sigma(\tau)}_\ell < 0$ at every point $p \in \sigma(\tau)$.
\end{proof}
\begin{corollary}
	Equation~\eqref{eq:SmallAreaLaw} also holds for past holographic screens.
\end{corollary}
\begin{proof}By Corollary IV.5 of Ref.~\cite{Bousso:2015qqa}, $\alpha > 0$ everywhere on past holographic screens, and $\theta^{\sigma(\tau)}_k = 0$ while $\theta^{\sigma(\tau)}_\ell > 0$ everywhere on the strictly marginally antitrapped leaves foliating a past holographic screen. The positive definiteness of the right-hand side is thus ensured.
\end{proof}

	Although the local area law in Corollary IV.4 of Ref.~\cite{Bousso:2015qqa} only applies to complete screen leaves $\sigma(\tau)$ and not surfaces of $H$ that are arbitrary deformations away from a leaf, here we are saying that Eq.~\eqref{eq:SmallAreaLaw} applies to proper area elements $\delta A$ at any point in $H$ since the foliation of $H$ by $\sigma(\tau)$ is unique by Theorem IV.8 of Ref.~\cite{Bousso:2015qqa}, and so every point lies on a leaf that is a marginal surface.

\subsubsection{Constructing $Y_{\star}$ by a local deformation of $Y$}
\label{sssec:Ystar}

Equipped with the results above, we now prove the existence of a surface $Y_{\star}$ on $H$ with $A[Y_{\star}] < A[Y]$.
\begin{proposition} \label{thm:Ystar}
	There exists a compact, boundaryless surface $Y_{\star}$ that is a complete cross section of $H$ that satisfies $A[Y_{\star}] < A[Y]$ where $Y$ is the intersection of $H$ with a null congruence of $X$, i.e., $Y = H \cap N_{\pm u}[X]$ for $u = \tilde{k}, \tilde{\ell}$. The map along integral curves of $h$ from $Y$ to $Y_\star$ is in the $-h$ direction.
\end{proposition}
\begin{proof}
	By Lemma~\ref{lem:Normals}, there is a point $x$ on both $Y$ and $\sigma_{\star}$ such that their normals are parallel.
	We can construct a new surface $Y_{\star}$ by deforming $Y$ at $x$ in the $-h$ direction by an infinitesimal length $\delta\tau$.
	Then, by Lemma~\ref{lem:SmallAreaLaw}, the proper area element strictly decreases, $\delta A(x-\delta\tau) < \delta A(x)$. The change in area associated with the tilt of this area element---that is, from the angle between the normals of $Y$ and $Y_\star$---is second-order in the angle, and hence second-order in $\delta\tau$, and can be dropped. 
	We therefore conclude that $Y_{\star}$ has strictly smaller area than $Y$, $A[Y_{\star}] < A[Y]$.
\end{proof}

This result justifies the assumption made in step~\ref{step1} in the chain of inequalities introduced in the beginning of Sec.~\ref{ssec:Nonspherical}.
The sequestration of holographic screens as argued by employing forbidden loops is therefore robust even in nonsymmetric spacetimes where the intersection $Y$ has nonzero Euler characteristic. 
We are now ready to give a precise statement of sequestration for general spacetimes.
\begin{theorem}[Sequestration] \label{thm:Sequestration}
For globally hyperbolic spacetimes satisfying the NEC, regular $C^2$-smooth holographic screens are forbidden from passing through the null congruences of an HRT surface $X$ with nonzero Euler characteristic. That is, the holographic screens are sequestered to live within the (closure of) one of the four causal wedges of $X$: the inner wedge, outer wedge, causal past, or causal future of $X$.
\end{theorem}
\begin{proof}
The surface $Y$ is homologous to $X$, since both are boundary-homologous, so the Euler characteristic of $Y$ is nonzero.	The conclusion then immediately follows from steps~\ref{step1}-\ref{step5} and the intervening results introduced throughout Sec.~\ref{ssec:Nonspherical}.\footnote{The proof for the special case of spherically symmetric spacetimes is subsumed in the case of $Y=\sigma$, i.e., when the intersection is a leaf of $H$, in which case forbidden loops can be drawn as in Fig.~\ref{fig:ForbiddenLoops}.}
\end{proof}

\subsection{Edge cases: Null congruences containing a complete slice of sequestered screens}
\label{ssec:EdgeCases}

The sequestration result derived above explicitly rules out the possibility of holographic screens crossing through null congruences of $X$. 
However, it is natural to ask if a screen $H$ that is sequestered to live within a wedge of $X$ is allowed to ``touch'' a null congruence of $X$, that is, to intersect $N_{\pm u}[X]$ without passing through to another wedge.

A key component of our strategy thus far has relied on comparing areas along the screen $H$. 
If $H$ does not pass through a null congruence of $X$, then in general, the intersection between $H$ and a null congruence $N_{\pm u}[X]$ does not have to be a ``complete slice'' of $H$.
By complete slice, we mean that the intersection $Y = H \cap N_{\pm u}[X]$ is such that $H$ intercepts each of the null generators of $N_{\pm u}[X]$ (see Fig.~\ref{fig:DeformProof}).
A complete slice of $H$ would necessarily be a compact, boundaryless surface by construction, but a partial slice will generally have a boundary.
If $Y$ were just a partial slice, then the existence of a surface $Y_{\star}$ with strictly smaller area than $Y$ is no longer ensured by Proposition~\ref{thm:Ystar} since the Poincar\'{e}-Hopf theorem (Theorem~\ref{thm:Index}) used to prove Lemma~\ref{lem:Normals} does not hold for manifolds with boundary (unless we enforce additional conditions on the vector field that cannot be guaranteed for our $v$).
It is therefore not guaranteed that we can use rule~\ref{rule2} or step~\ref{step1} of the forbidden loop toolkit to compare areas along $H$ when $Y$ is a partial intersection.
Even if we suppose for a moment that the surface $Y_{\star}$ exists, since the surface $N_{\pm k}[Y_{\star}] \cap \Sigma$ in step~\ref{step3} is not Cauchy-splitting, we can no longer use rule~\ref{rule3} or step~\ref{step3} to compare its area with the HRT surface $X$.
Thus, $Y = H \cap N_{\pm u}[X]$ being a partial intersection presents an obstruction to finding forbidden loops as it renders the area inequalities of rules~\ref{rule2} and \ref{rule3} (in the symmetric case) or steps~\ref{step1} and \ref{step3} (in the nonsymmetric case) invalid.
In what follows, we will therefore restrict ourselves to intersections $Y$ that are complete slices of $H$.
We find that although the answer depends subtly on the null congruence in question, we can nevertheless forbid screens from touching any null congruence when their intersection $Y$ is a complete slice of $H$.

\subsubsection{$N_{-\tilde{\ell}}[X]$ cannot contain a complete slice of $H$}\label{sec:extraref}
Suppose that $H$ is in either the exterior or causal past of $X$, and that the intersection of $H$ with the null congruence $N_{-\tilde{\ell}}[X]$ is a complete slice $Y$, i.e., $H \cap N_{-\tilde{\ell}}[X] = Y$.
We would like to compare ingoing, future-directed null expansions at $X$ and $Y$, by adopting a similar strategy to the proof of Lemma~\ref{lem:Normals} and invoking the Poincar\'{e}-Hopf theorem.
For every point $p \in Y$, identify a leaf $\sigma_p$ of the screen $H$ that contains $p$.
From that leaf, take the ingoing future-directed null generator of congruences orthogonal to $\sigma_p$ and call it $\ell(p)$.
This defines a null expansion $\theta_{\ell}(p)$ everywhere along $N_{+\ell}[Y]$, by which we mean the pointwise-defined null congruence as elaborated on in the discussion of $N_{\pm k}[Y_{\star}]$ for step~\ref{step4}.
Since by Theorem~IV.8 of Ref.~\cite{Bousso:2015qqa} each leaf of $H$ is strictly marginally trapped, each $\theta_{\ell}(p)$ at $p \in Y$ is strictly negative.
Furthermore, by focusing~\eqref{eq:CIDF}, each $\theta_{\ell}(p)$ will remain negative to the future along $N_{+\ell}[Y]$, and thus on every point of $X$.
Now, take $\tilde{\ell}$ to be the ingoing, future-directed null generator of congruences orthogonal to $X$.
Note that because $Y$ is not a constant-affine slice of $N_{-\tilde{\ell}}[X]$, $\ell(p)$ will in general not point in the same direction as $\tilde{\ell}$. 
Since $\tilde{\ell}$ is by definition orthogonal to $X$, we can once again define a vector field $v$ as the projection of $\ell(p)$ onto $X$ much like the setup shown in Fig.~\ref{fig:HairyBall}.
More precisely, $v$ will be the projection of $\ell(p)$ onto the tangent space of points on $X$ that the pointwise-defined null congruence $N_{+\ell}[Y]$ intersects.
By the Poincar\'{e}-Hopf theorem (Theorem~\ref{thm:Index}), $v$ has to vanish at some point, say $y$, on $X$, if $\chi(X) \neq 0$.
At that point, we have that $\theta_{\ell}(y)$ and $\theta_{\tilde{\ell}}$ are related by an affine scaling since $\ell(y)$ and $\tilde{\ell}$ are parallel to each other.
However, since $\theta_{\ell}(p) < 0$ for all $p \in Y$, and can only grow more negative on $N_{+\ell}[p]$, but $\theta_{\tilde{\ell}} = 0$ everywhere on $X$ by definition of an extremal surface, we have a contradiction.
As a consequence, $N_{-\tilde{\ell}}[X]$ cannot contain a complete slice $Y$ of a future holographic screen. 

As alluded to in Fig.~\ref{fig:ForbiddenLoops}, the forbidden loop for a screen passing through $N_{-\tilde{\ell}}[X]$ did not rely on the defining Cauchy slice $\Sigma$ for a maximin HRT surface $X$.
It is now apparent that since our argument only depends on properties of ingoing future-directed null expansions along $N_{-\tilde{\ell}}[X]$ and the fact that future holographic screens are foliated by strictly marginally trapped surfaces, it is irrelevant whether the screen approaches $N_{-\tilde{\ell}}[X]$ from the exterior or past of $X$.
This statement holds for both the earlier analysis of a screen crossing $N_{-\tilde{\ell}}[X]$, as well as the current analysis of $N_{-\tilde{\ell}}[X]$ containing a complete slice of $H$.
As such, our sequestering results for $N_{-\tilde{\ell}}[X]$ are valid without ever invoking forbidden loops.

\subsubsection{$N_{+\tilde{\ell}}[X]$ cannot contain a complete slice of $H$}
Next, suppose that $H \cap N_{+\tilde{\ell}}[X] = Y$ is a complete slice of a screen $H$ in either the interior or causal future of $X$.
In this scenario, it is easy to identify the forbidden loops that lead to contradictory area inequalities as shown in Fig.~\ref{fig:TouchTopLeft}.

\begin{figure}[t]
	\centering
	\includegraphics[width=\linewidth]{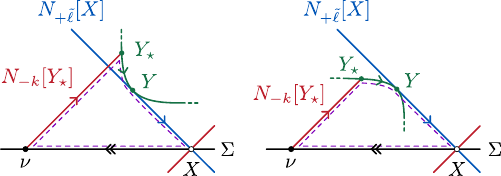}
	\caption{Forbidden loops (purple dotted lines) appearing when the null congruence $N_{+\tilde{\ell}}[X]$ contains a complete slice of a future holographic screen in the future of $X$ (left panel), or the interior of $X$ (right panel). For the special case of symmetric spacetimes, simply make the substitution $Y_{\star} \rightarrow \sigma$ and $Y \rightarrow \sigma'$.}
	\label{fig:TouchTopLeft}
\end{figure}

The key observation here is that regardless of whether $H$ is in the interior or future of $X$, in spherically symmetric spacetimes there always exists another leaf $\sigma \in H$ that flows to $\sigma'$ with strictly increasing area $A[\sigma] < A[\sigma']$, whereas in nonsymmetric spacetimes, $Y$ being a complete slice of $H$ implies the existence of $Y_{\star}$ with $A[Y_{\star}] < A[Y]$ by Proposition~\ref{thm:Ystar}.
The past null congruences $N_{-k}[\sigma]$ and $N_{-k}[Y_\star]$ fired from $\sigma$ and $Y_{\star}$ respectively both intersect the Cauchy slice $\Sigma$ at $\nu$ and close the loop.
This allows us to close the chain of inequalities in the same fashion as Figs.~\ref{fig:ThreeRules} and \ref{fig:ForbiddenLoops}.
We therefore also conclude that $N_{+\tilde{\ell}}[X]$ cannot contain a complete slice of any future holographic screen.

\subsubsection{$N_{\pm \tilde{k}}[X]$ cannot contain a complete slice of $H$}

The final configurations to consider are, first, when $H \cap N_{+\tilde{k}}[X] = Y$ with $H$ in the exterior or future of $X$, and second, when $H \cap N_{-\tilde{k}}[X] = Y$ with $H$ in the interior or past of $X$.
For a screen sequestered to the future of $X$ to intersect $N_{+\tilde{k}}[X]$ at $Y$, the screen must continue from $Y$ to $X$ along $N_{+\tilde{k}}[X]$ to preserve monotonicity of cross-sectional area.
Likewise, a screen sequestered to the interior of $X$ intersecting $N_{-\tilde{k}}[X]$ at $Y$ must extend from $Y$ to $X$ along $N_{-\tilde{k}}[X]$.
In both of these cases, it is straightforward to find forbidden loops as demonstrated in Fig.~\ref{fig:LoopsNk}.
\begin{figure}[h]
	\centering
	\includegraphics[width=0.7\linewidth]{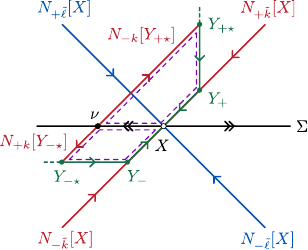}
	\caption{Forbidden loops (purple dotted lines) that arise for sequestered holographic screens approaching $Y_{\pm}$ on $N_{+\tilde{k}}[X]$ and $N_{-\tilde{k}}[X]$ from the future and interior, respectively. Similar loops exist when the screens approach $Y_{\pm}$ in a spacelike and timelike fashion, respectively.}
	\label{fig:LoopsNk}
\end{figure}
We can think of these as limits of screens depicted in Fig.~\ref{fig:HRT_EndLoops} where a finite segment of screens ending on $X$ are boosted onto $N_{\pm \tilde{k}}[X]$.

On the other hand, for a screen sequestered to the exterior of $X$ to simultaneously intersect $N_{+\tilde{k}}[X]$ at $Y$ and respect monotonicity of area along the screen, the screen must have a semi-infinite null segment with area flow from future null infinity to $Y$ along $N_{+\tilde{k}}[X]$, then escape from $Y$ to the outer wedge of $X$.
In the same vein, a screen sequestered to the past of $X$ will necessarily have a null segment with flow from past null infinity to $Y$ along $N_{-\tilde{k}}[X]$, and leave $Y$ to the past of $X$.
No forbidden loops exist for such outgoing screens from $N_{\pm \tilde{k}}[X]$, so these constructions (shown in Fig.~\ref{fig:TouchNk}) are in principle allowed.
While these are still holographic screens in a general sense, the presence of a semi-infinite null segment means that they run afoul of a genericity condition in Definition II.8 of Ref.~\cite{Bousso:2015qqa} that defines \emph{regular} holographic screens.
As reviewed in Sec.~\ref{sec:BackgroundDefs}, this condition requires a regular holographic screen to have no set of leaves of nonzero measure on which the tangent to the screen's flow is null and outward-directed.
If we restrict ourselves to only regular holographic screens, then we are forced to forgo the remaining two types of irregular screens containing semi-infinite null segments along $N_{\pm \tilde{k}}[X]$.

An interesting special case of these ``irregular'' screens arises under symmetry (e.g., spherical), when the intersection itself is a leaf of the screen.
When $Y=\sigma$, since by definition $\theta_k = 0$ on $\sigma \in H$, focusing by the Raychaudhuri equation says that if $N_{+k}[X]$ is to contain $\sigma$ as a cross section, the segment of $N_{+k}[X]$ between $X$ and $\sigma$ must be a stationary light sheet, i.e., the segment must have vanishing null expansion $\theta_k = 0$, giving $A[X] = A[\sigma]$.
The configuration with a stationary light sheet on a segment of $N_{+k}[X]$ and a spacelike screen leaving it to the exterior is precisely the Engelhardt-Wall construction of Refs.~\cite{Engelhardt:2017aux,Engelhardt:2018kcs}.
However, our results excluding irregular screens is not in contention with the construction of Refs.~\cite{Engelhardt:2017aux,Engelhardt:2018kcs}, as the goal there was not to construct a holographic screen.
Rather, by deleting everything but the outer wedge of a minimar leaf $\mu$ on a spacelike segment and gluing a past-directed stationary light sheet to the interior of $\mu$ such that $A[X] = A[\mu]$, Engelhardt and Wall sought to build a partial Cauchy slice that guarantees the existence of a maximin HRT surface $X$ for the outer entropy.

\begin{figure}[t]
	\centering
	\includegraphics[width=0.7\linewidth]{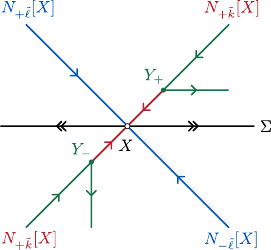}
	\caption{Irregular screens that intersect $N_{\pm \tilde{k}}[X]$ at $Y_{\pm}$.}
	\label{fig:TouchNk}
\end{figure}

In summary, we have found that not only are regular holographic screens $H$ sequestered to live within one of four causal wedges defined by the null congruences of $X$, such sequestered screens are further forbidden from intersecting $N_{\pm u}[X]$ for $u=\tilde{k}, \tilde{\ell}$ in such a way that the intersection is a complete slice of $H$.

\section{Application: constraints on coarse graining}
\label{sec:Application}

As an illustration of the results of the previous section, we now turn to the question of whether and how one could give a coarse grained entropic interpretation to the areas of leaves of holographic screens that are \emph{not} apparent horizons.
Specifically, we mean holographic screens whose leaves are not stable, outermost, marginal surfaces for which the corresponding outer wedges do not nest.
In spherically symmetric spacetimes, these are screens that have at least one nonempty interval $[\tau_1, \tau_2]$ on which the tangent $(\partial_\tau)^\mu = h^\mu$ is everywhere timelike.
Such screens are ubiquitous in cosmology.
Already in AdS/CFT, a holographic screen need not be foliated by minimar surfaces, and it turns out that the existing proof of \Eq{eq:Sout_area} fails to hold in such cases.
It would be illuminating to uncover an analogous relation for such screens, and here we investigate how screen sequestration can inform analogous constructive attempts at a proof.

First, to understand the technical reasons why \Eq{eq:Sout_area} ceases to hold, let us give a rough sketch of its proof before considering a concrete example.
The proof of Eq.~\eqref{eq:Sout_area} essentially proceeds in two steps~\cite{Engelhardt:2017aux,Engelhardt:2018kcs}.
First, assuming that a spacetime contains both a minimar surface $\mu$ and an HRT surface $X$ that is homologous to the same boundary as $\mu$, one shows that $A[X] \leq A[\mu]$.
In other words, the area of a minimar surface provides an upper bound for the area of any HRT surface for $B$ and hence also an upper bound for $S(\rho_B)$ via the HRT formula.
It therefore follows that $S^{\text{(outer)}}[\mu] \leq A[\mu]/4G\hbar$.
Second, one saturates the bound by constructing a spacetime that contains both $O_W[\mu]$ and an HRT surface $X_\star$ that is homologous to $\mu$ and has $A[X_\star] = A[\mu]$, thereby proving \Eq{eq:Sout_area}.

The first step is relatively straightforward, with the bound following as a consequence of focusing (the Raychaudhuri equation \eqref{eq:CIDF}).
The second step is more involved and requires that several subtle points about existence and uniqueness be addressed, which we will not review here; nevertheless, we describe the basics of the construction below (see also Fig.~1 of Ref.~\cite{Engelhardt:2017aux} for illustration).
Let $k$ denote the outward-pointing, future-directed, marginal null direction on $\mu$ in which $\theta_k = 0$ and let $\ell$ denote the direction in which $\theta_\ell < 0$.
(Here we have in mind that $\mu$ is a leaf of a future holographic screen, for which the area of leaves increases toward the past and toward the AdS boundary; the case of a past holographic screen is just the time reverse.)
Beginning with $O_W[\mu]$ held fixed, one first glues a stationary null light sheet (i.e., with $T_{kk}$ and $\varsigma_k$ both vanishing) to $\mu$ in the $-k$ direction.
The expansion $\theta_k$ remains zero on the light sheet by virtue of it being stationary, and so the light sheet is foliated by constant-area slices.
Furthermore, the condition for strict spacetime stability, $\nabla_k \theta_\ell < 0$, continues to hold on the light sheet.
Equivalently, we may write $\nabla_{-k} \theta_\ell > 0$.
Therefore, moving along the light sheet away from $\mu$ in the $-k$ direction, one eventually encounters\footnote{In general, for the nonspherical case, the constructed surface $X_\star$ on which $\theta_k=\theta_\ell=0$ is not a surface of constant affine parameter along $N_{-k}$, and hence is not extremal, since $\ell$ is not surface-orthogonal. Nevertheless, given $X_\star$ one can prove that a bona fide HRT surface exists as a section of $N_{-k}$, as was done in Ref.~\cite{Engelhardt:2018kcs}.}
 an extremal slice $X_\star$ on which $\theta_k = \theta_\ell = 0$.
Labeling the segment of the stationary light sheet from $\mu$ to $X_\star$ by $N_{-k}[\mu; X_\star]$, one then glues the CPT conjugate of $O_W[\mu] \cup N_{-k}[\mu; X_\star]$ to $X_\star$, thereby engineering complete Cauchy data for a spacetime with two asymptotic boundaries.
The extremal surface $X_\star$ constructed in this way ends up being an HRT surface for the spacetime (i.e., one can show that it is an extremal surface of \emph{minimal} area) with the same area as $\mu$, which completes the proof. 

We reviewed a sketch of Engelhardt and Wall's construction in Refs.~\cite{Engelhardt:2017aux,Engelhardt:2018kcs} in order to point out a small yet crucial detail; namely, that the construction relies on the stability property of minimar surfaces, $\nabla_k \theta_\ell < 0$.
However, it is easy to find examples of marginal surfaces that do not obey strict spacetime stability and for which one might expect to still have an entropic interpretation.
In the interior of an AdS black hole formed from the gravitational collapse of matter, for instance, the leaves of a future holographic screen can fail to obey strict spacetime stability inside of the collapsing matter, and are as such not minimar.
In other words, the leaves remain closed, marginally trapped, and boundary-homologous, but they fail to be minimar because the sign of $\nabla_k \theta_\ell$ is wrong.

In fact, by a very mild extension of a result of Booth et al.~\cite{Booth:2005ng}, at least in the spherically symmetric case it follows that the sign of $\nabla_k \theta_\ell$ on a leaf $\sigma$ directly correlates with the character of the screen at that leaf, regardless of whether it is a future or past holographic screen.\footnote{Booth et al.~\cite{Booth:2005ng} were concerned with marginally trapped tubes, but the arguments apply just as well to the marginally antitrapped case and to holographic screens.}
Recalling the notation from \Sec{ssec:ScreenDefs}, let $h^\mu$ denote the tangent vector to a holographic screen $H$ in the direction of increasing area.
Let us choose a normalization for $k$ such that $k \cdot \ell = -1$ and $h^\mu = \alpha \ell^\mu + \beta k^\mu$ where $\alpha < 0$ or $>0$ for future and past holographic screens, respectively.
Since $k$ is the marginal direction, it follows that the Lie derivative $\mathcal{L}_h \theta_k = 0$, from which we obtain
\begin{equation}
-\frac{\alpha}{\beta} = \frac{\mathcal{L}_k \theta_k}{\mathcal{L}_\ell \theta_k} = \frac{\nabla_k \theta_k}{\nabla_\ell \theta_k}.
\end{equation}
Furthermore, since $h_\mu h^\mu = -2\alpha\beta$, the character of $h^\mu$ fixes the sign of $\beta$.
In the spherically symmetric case, for which $h^\mu$ has the same character everywhere on a single leaf, $\alpha/\beta$ is therefore positive when $H$ is timelike, negative when $H$ is spacelike, and zero when $H$ is null.
Focusing implies that $\nabla_k \theta_k \leq 0$, and in the spherically symmetric case, cross-focusing (see \Eq{eq:CIDF}) implies that $\nabla_\ell \theta_k = \nabla_k \theta_\ell$.
We therefore conclude that, for a spherically symmetric holographic screen $H$, the sign of $\nabla_k \theta_\ell$ at a location on $H$ is in direct correspondence with the character of $H$ at that location:
\begin{equation}
\nabla_{k} \theta_\ell = \left\{
\begin{array}{ll}
> 0 & H~\text{timelike} \\
0~\text{or undefined} & H~\text{null} \\
< 0 & H~\text{spacelike}.
\end{array} \right.
\end{equation}

It would be somewhat surprising if the interpretation of holographic screen leaf area as a coarse grained entropy were to fail in AdS/CFT where a future holographic screen transitions to being a nonspacelike hypersurface.
With this perspective in mind, let us continue to examine the timelike portion of a holographic screen in the interior of a holographic black hole with the additional assumption that the spacetime is spherically symmetric, for simplicity.
Consider a leaf $\sigma$, and let us attempt to mirror the construction of Refs.~\cite{Engelhardt:2017aux,Engelhardt:2018kcs} as closely as possible.
On such a leaf, we have that $\theta_k = 0$, $\theta_\ell < 0$, and $\nabla_k \theta_\ell > 0$.
Therefore, based only on the signs of expansions and cross-focusing, the analogous construction is to join a stationary null light sheet in the $+k$ direction, so that one of its future slices is eventually extremal.
Let us label this slice $X$.

Now, it is unclear what region relative to $\sigma$ should be held fixed under a coarse graining prescription.
Let us explore one option based on plausible assumptions; however, the choices described below are not necessarily unique.
The area law obeyed by holographic screens suggests that the regions corresponding to different $\sigma$ should nest along the screen.
Together with the fact that we join a null light sheet to $\sigma$ along the $+k$ direction, this suggests that the past wedge of $\sigma$ is what should be thought of as being held fixed.
However, according to this reasoning based on the area law, it seems reasonable that the region being held fixed should be augmented to include the part of the holographic screen that is foliated by leaves with a larger area than $\sigma$.
For future reference, let us call this the part of the screen that is forward to $\sigma$.

At this point we can invoke the sequestration arguments of the previous section.
Consider the null congruences that emanate from $X$.
So that we avoid a situation in which the congruence in the $-\ell$ direction intersects the holographic screen---as discussed in Sec.~\ref{sec:extraref}---we conclude that there must be some obstruction to the congruence before it hits the screen, for example, in the form of a singularity.
See \Fig{fig:timelike_construction} for illustration.
(Of course, even such a case would correspond to an irregular holographic screen, cf. Fig.~\ref{fig:TouchNk}.)

\begin{figure}
\includegraphics[width=0.9\columnwidth]{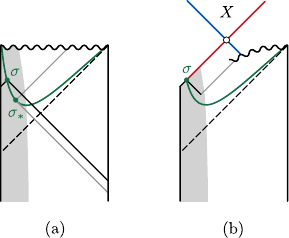}
\caption{(a) Future holographic screen inside a spherically symmetric AdS black hole formed from the gravitational collapse of matter (shaded). In the matter-containing region, the screen can have timelike character, and past wedges of leaves in this timelike part nest.
However, if we insist that the region that is fixed under coarse graining nest with all such forward regions, then at minimum we must also take the union of the past wedge of $\sigma$ with the outer wedge of $\sigma_*$, which is where the screen goes from being timelike to spacelike.
(b) Upon attempting to mimic the construction of Engelhardt and Wall~\cite{Engelhardt:2017aux,Engelhardt:2018kcs}, we see that any constructed extremal surface $X$ should be screened, e.g., by a singularity, so as to avoid running afoul of sequestration constraints.}
\label{fig:timelike_construction}
\end{figure}

Alternatively, the coarse graining procedure could simply not include the part of the holographic screen that is forward to $\sigma$, but then the challenge is to make sense of the area law.
The observation that a holographic screen can alternate between having timelike and spacelike components poses a further conceptual challenge to such an approach.
The basic construction of Refs.~\cite{Engelhardt:2017aux,Engelhardt:2018kcs} goes through whenever the screen is spacelike, but it becomes unclear how to bridge the coarse graining procedure for successive spacelike components through timelike components.
In any case, these observations illustrate the ways in which screen sequestration can inform possible coarse graining prescriptions.
We also remark that understanding how to coarse grain the leaves of nonspacelike holographic screens in AdS black holes may suggest ways to coarse grain holographic screens in cosmology.
This latter problem is doubly puzzling, as cosmological models are bereft of an entropic formula like the HRT prescription, short of embedding a cosmology in a holographic spacetime~\cite{Antonini:2022blk,Sahu:2023fbx}.

\section{Discussion}
\label{sec:Discussion}

We end this article by connecting the results presented above to existing literature, as well as to possible generalizations and future research directions.

\subsection{Relation to other works}

Holographic screen sequestration by HRT surfaces most closely connects with recent work by Engelhardt and Folkestad, in which they proved that trapped surfaces in the classical limit of AdS/CFT must lie behind event horizons \cite{Engelhardt:2020mme}.
As an intermediate step, they showed that minimar surfaces must lie behind event horizons, and thus so do spacelike holographic screens.
In AdS/CFT, a consequence of causal wedge inclusion is that any HRT surface for a complete connected component of the asymptotic boundary must be causally disconnected from the boundary component in question.
If the HRT surface is nontrivial, it must therefore lie behind an event horizon.
Therefore, in the case where a minimar surface is part of a holographic screen that is hidden behind an event horizon that also hides an HRT surface, we can think of sequestration as further constraining the region behind the horizon in which the minimar surface can lie.
In another sense, sequestration is complementary to Engelhardt and Folkestad's findings, since it also constrains where nonminimar surfaces---namely, marginal surfaces that foliate timelike, mixed, or indefinite-signature holographic screens---can lie when HRT surfaces are present.
The arguments for screen sequestration do not make use of the holographic dictionary, and so in principle the results remain applicable for spacetimes that are not asymptotically AdS, provided that they admit nontrivial (boundary-homologous) HRT surfaces.

\subsection{Other extremal surfaces}

In proving results about sequestration, we worked exclusively with HRT surfaces, which in particular are extremal surfaces that are maximin.
There are of course other types of extremal surfaces: minimax extremal surfaces, which have maximal area on a Cauchy surface and have the least area among the maximal area slices of a complete family of Cauchy surfaces, or even ``maximax'' surfaces, which are largest area surfaces among such a family of maximal area surfaces.
Examples of minimax and maximax extremal surfaces are the equatorial sphere of the throat of a closed bouncing Friedmann--Lema\^{i}tre--Robertson--Walker (FLRW) cosmology and the equatorial sphere of the maximal volume slice of a big bang--big crunch FLRW cosmology, respectively.
In such cases, however, the extremal surface has little sequestering power.
In a diagram that is analogous to \Fig{fig:ForbiddenLoops}, the extremal surface $X$ would look like an attractor rather than a saddle.
The absence of any lines that represent surfaces of increasing area away from $X$ prevent us from forming the forbidden loops that were necessary to argue for sequestration, save for the loop shown in the bottom right of \Fig{fig:ForbiddenLoops}.
We therefore conclude that future (past) holographic screens are only forbidden from crossing from $O_W[X]$ into $I^-[X]$ ($I^+[X]$) in the nonmaximin case. 

\subsection{Semiclassical generalizations}

In this article, we presented purely geometrical arguments and demonstrated results that apply to classical general relativity.
Nevertheless, we expect the results on sequestration to continue to go through when semiclassical quantum corrections are taken into account.

In particular, one can define a semiclassical version of a holographic screen called a \emph{Q-screen} \cite{Bousso:2015eda}, essentially by replacing all instances of references to the area of surfaces with the generalized entropy associated to surfaces~\cite{Bekenstein:1972tm,Bekenstein:1973ur,Hawking:1976de}.
Given a closed, spacelike, codimension-two surface $\sigma$ that splits a Cauchy hypersurface $\Sigma$ into the interior $\Sigma^-[\sigma]$ and exterior $\Sigma^+[\sigma]$, the generalized entropy associated with $\sigma$ is
\begin{equation}
S_\mrm{gen}[\sigma] = \frac{A[\sigma]}{4 G\hbar} + S(\rho_\mrm{out}),
\end{equation}
where $\rho_\mrm{out} = \tr_\mrm{\Sigma^-[\sigma]} \rho$ and $\rho$ denotes the state of quantum fields on $\Sigma$.
A Q-screen is then a codimension-one hypersurface that admits a foliation by surfaces whose generalized entropy is stationary with respect to deformations in one null direction and has variations of definite sign in the other null direction.
Equivalently, one can define the \emph{quantum expansion} in the null direction $k$ at a point $y \in \sigma$ as
\begin{equation}
\Theta_k[\sigma; y] = \lim_{\delta A \rightarrow 0} \left. \frac{4 G\hbar}{\delta A} \frac{d S_\mrm{gen}}{d\lambda} \right|_y,
\end{equation}
where $\delta A$ is a small element of area on $\sigma$ at $y$ and $\lambda$ is an affine parameter for the null geodesic generated by $k$.
The leaves of a future (respectively, past) Q-screen then satisfy $\Theta_k[\sigma] = 0$ and $\Theta_\ell[\sigma] < 0$ (respectively, $>0$) where, just as with (classical) expansion, in omitting from the argument the point $y \in \sigma$ we mean that the (in)equality holds everywhere on $\sigma$.

In complete analogy with holographic screens, one can show that Q-screens obey a monotonic generalized entropy law \cite{Bousso:2015eda} if one assumes the quantum focusing conjecture (QFC) \cite{Bousso:2015mna}, which states that quantum expansion is nondecreasing along any null congruence,
\begin{equation}
\frac{d \Theta}{d \lambda} \leq 0.
\end{equation}

In order to generalize sequestration to Q-screens, the last generalization that we must make is to replace maximin HRT surfaces with quantum maximin surfaces \cite{Akers:2019lzs}.
Analogously to classical maximin surfaces, a quantum maximin surface is obtained by minimizing generalized entropy on a complete family of Cauchy hypersurfaces and then identifying the minimal surface of maximum generalized entropy (or any one of such surfaces if there are several).
A quantum maximin surface has vanishing quantum expansions, and it can be shown that a quantum maximin surface is also a quantum extremal surface \cite{Engelhardt:2014gca} provided that it possesses an additional stability property \cite{Akers:2019lzs}.
Assuming the QFC, it then follows that quantum maximin surfaces sequester Q-screens in the same way that classical maximin surfaces sequester holographic screens.

In connection with semiclassical generalizations, it is natural to ask whether the results presented in this work have any connection to quantum extremal islands (see Ref.~\cite{Almheiri:2020cfm} for a review). 
These regions appear in the islands formula, which is a proposal for how to compute the (fine grained) von~Neumann entropy of a nongravitating system when it is coupled to an auxiliary system that gravitates.
Among its by-now myriad applications, the islands formula has been shown to produce a Page curve for the Hawking radiation in a large variety of models of evaporating black holes coupled to a nongravitating reservoir (e.g., Refs.~\cite{Penington:2019npb,Almheiri:2019hni,Akers:2019nfi,Penington:2019kki,Almheiri:2019psy,Gautason:2020tmk,Hollowood:2020cou} to name but a few).
If $R$ is the nongravitating system, then the islands formula gives the following prescription for computing the von~Neumann entropy of the reduced state of $R$:
\begin{equation}
S(\rho_R) = \min \ext_I \left\{ S(\tilde \rho_{R \cup I}) + \frac{A[\partial I]}{4 G\hbar}  \right\}.
\end{equation}
The extremization and minimization is over the location and shape of an \emph{island} $I$---in other words, a (partial) Cauchy slice---within the gravitating system, and the tilde over $\tilde \rho_{R \cup I}$ is to denote that we are instructed to evaluate the von~Neumann entropy of the \emph{semiclassical} state of quantum fields on the curved spacetime region $R \cup I$.

While $S(\rho_R)$ is a generalized entropy and the boundary of the island is a quantum extremal surface, we cannot immediately conclude that $\partial I$ sequesters Q-screens in the same way that a quantum maximin surface does, since the region $R \cup I$ need not have any relation to the region $\Sigma^+[\sigma]$ used in defining the Q-screen.
That said, it may be possible to define a Q-screen in a different way, e.g., by replacing $\rho_\mrm{out} = \tr_{\Sigma^-[\sigma]}\rho$ with a state defined on $\Sigma^-[\sigma] \cup R$ or some other region that includes $R$, for which conclusions may be drawn about sequestration by $\partial I$.

In relation to \Sec{sec:Application}, quantum extremal islands may have some connection to coarse grainings based on leaves of timelike holographic screens in the interior of an AdS black hole.
For instance, an island in the black hole interior provides a new geometric reference that could affect the way one thinks about nesting of the causal wedges of screen leaves.
While these comments are perhaps the seeds of interesting questions, we leave further investigation to future work.

~

\noindent {\it Acknowledgments:}  We thank Bartek Czech, Edgar Shaghoulian, and Mark Van Raamsdonk for comments and discussions during the preparation of this manuscript.
A. C.-D. was supported for a portion of this work by the Natural Sciences and Engineering Research Council of Canada (NSERC), [funding reference number PDF-545750-2020], and by Grant \#62312 from the John Templeton Foundation for ``The Quantum Information Structure of Spacetime'' (QISS) project.
P.L. is supported by the Four Year Doctoral Fellowship at the University of British Columbia.
G.N.R. is supported by the James Arthur Postdoctoral Fellowship at New York University. This work has benefited from discussions at the QIMG 2023 long term workshop (YITP-T-23-01) held at the Yukawa Institute for Theoretical Physics, Kyoto University.

\newpage

\bibliographystyle{utphys-modified}
\bibliography{holo-screen-sequestration.bib}

\end{document}